\newif\ifFull
\Fulltrue

\ifFull
    \documentclass[11pt]{article}
    \usepackage[margin=1in]{geometry}
\else
    \documentclass[envcountsame]{llncs}  
    \renewcommand{\subsection}[1]{\paragraph{\bf #1.}}
\fi
\ifFull 
    \usepackage{amsthm}
\else
    \usepackage{times}
\fi
\usepackage{hyperref} 
\usepackage{amsmath,amsfonts,amssymb} 
\usepackage{graphicx,float} 
\usepackage{enumerate,paralist} 
\usepackage[ruled]{algorithm2e} 
\usepackage{cite}
\usepackage[capitalize]{cleveref}
\usepackage[utf8]{inputenc} 
\usepackage{lmodern} 
\usepackage{lineno}
\usepackage{microtype} 
\usepackage{multirow} 
\usepackage{afterpage} 

\ifFull
    \theoremstyle{definition}
    \newtheorem{theorem}{Theorem}
    \newtheorem{lemma}[theorem]{Lemma}
    \newtheorem{definition}[theorem]{Definition}
    \newtheorem{proposition}[theorem]{Proposition}
    
\else
    \let\doendproof\endproof
    \renewcommand\endproof{~\hfill\qed\doendproof}
\fi
\newtheorem{observation}[theorem]{Observation}
\crefformat{equation}{(\textup{#2#1#3})}
\Crefformat{equation}{Equation~(\textup{#2#1#3})}

\newcommand{\set}[1]{\left\{#1\right\}}
\newcommand{\varset}[2]{\left\{#1\ \middle|\ #2\right\}}
\newcommand{\abs}[1]{\left\lvert #1 \right\rvert}

\newcommand{\Round}[1]{\left( #1 \right)}
\newcommand{\Square}[1]{\left[ #1 \right]}

\newcommand{\defeq}{\stackrel{\text{def}}{=}}
\renewcommand{\emptyset}{\varnothing}
\renewcommand{\epsilon}{\varepsilon}
\newcommand{\prob}[1]{\mathbb{P}\Square{#1}} 
\newcommand{\expec}[1]{\mathbb{E}\Square{#1}} 

\newcommand{\keygen}{\mathsf{keygen}}
\newcommand{\Aux}{\mathsf{Aux}}
\newcommand{\dist}{\mathsf{dist}}
\newcommand{\wmark}{\mathsf{mark}}
\newcommand{\identify}{\mathsf{identify}}
\newcommand{\N}{\mathbb{N}} 
\newcommand{\R}{\mathbb{R}} 
\newcommand{\wlabel}{\mathsf{label}}
\newcommand{\id}{\mathsf{id}}
\renewcommand{\vector}[1]{\mathbf{#1}}
\newcommand{\hm}{i_{H}} 
\newcommand{\ml}{i_{M}} 

\title{Models and Algorithms for Graph Watermarking}
\author{David Eppstein$^1$ \hspace*{2em} Michael T. Goodrich$^1$ 
    \hspace*{2em}
    Jenny Lam$^2$ \\
    Nil Mamano$^1$
    \hspace*{2em}
    Michael Mitzenmacher$^3$ 
    \hspace*{2em}
    Manuel Torres$^1$
\\[4pt]
    $^1$Dept. of Computer Science, University of California, Irvine, CA USA \\
    $^2$Dept. of Computer Science, San Jos\'{e} State University, San Jos\'{e}, CA USA \\
    $^3$Dept. of Computer Science, Harvard University, Cambridge, MA USA 
}
\date{}

\newcommand{\proofErdosDegSep}{
We quantify and extend the probability analysis 
of a proof from~\cite{Bollobas1998}.
Let
\[
K = pn + (x - \epsilon)(pnq)^{1/2}, \quad \epsilon = \Round{\log(n/m)}^{1/2}.
\]
The event of the result fails if $d_m < K$ or if there is $i < m$ such that $d_i - d_{i+1} < K$.

The statement of theorem 3.12 of \cite{Bollobas1998} still holds when the words ``a.e. $G_p$ satisfies'' are replaced by ``$G_p$ satisfies with probability greater than $1-1/\omega(n)^2$''. This can be seen directly from the part of the proof where Chebychev's inequality is applied. 

By this result, the probability that $d_m < K$ is $1/\Square{m \Round{\log (n/m)}^2}$. The probability that $d_i - d_{i+1} < K$ for a given $i < m$ is $O(\alpha(n))$.
}

\newcommand{\proofLemErdosSep}{
We prove the theorem with probability at least $1 - (d+2)n^{-(1-\epsilon)/8}$.
Let $\alpha(n) = dn^{-(1-\epsilon)/4}$ and $m = h$. By \cref{lem:erdos-deg-sep}, the probability that $d_i - d_{i + 1} < d$ for some $i < h$ is at most $(d+1)n^{-(1-\epsilon)/8}$.

Let $X_{ij}$ be the expected neighborhood distance between two vertices $i,j \geq h$. We have 
\[
\expec{X_{ij}} = mp(1-p) \geq (2C + 1)n^{(1-9\epsilon)/8}\log n \geq (2C+1)\log n,
\]
so that, if $d' = C \log n$,
\[
\frac{(\expec{X_{ij}} - d')^2}{\expec{X_{ij}}} \geq \expec{X_{ij}} - 2d' \geq C \log n.
\]
Since the high-degree vertices are separated by more than two degrees, the fact that they are high-degree vertices is independent of whether they are neighbors of $i$ and $j$.
Consequently, we can apply a Chernoff bound (\cref{lem:chernoff}.)
Then, by the union bound, the probability that $X_{ij} < C \log n$ for some medium-degree $i, j$ is less than $n^{-C + 2} \leq n^{-(1-\epsilon)/8}$.
}

\newcommand{\proofLemIDSep}{
The expected distance between two such strings is at least $2\ell p(1 - p) \geq \ell p.$ 
Applying \cref{lem:chernoff} with $\lambda = \ell p/2$, we have that the probability that their Hamming distance is less than $\ell p/2$ is at most $e^{-\ell p/8}\leq n^{-(2C+C')}$. Therefore, the probability that at least two out of $k$ strings are within Hamming distance $D \leq \ell p/2$ of each other is at most $k^2n^{-(2C+C')} = n^{-C'}$.
}

\newcommand{\proofLemGuess}{
In the process of selecting $\ell$ edges without replacement, let $A$ be the event that the sample contain at least $R$ red edges, and let $B$ be the event that the sample satisfies the degree constraint. The event whose probability we want to bound is equal to
\[
  \prob{A|B} \leq \frac{\prob{A}}{\prob{B}}.
\]

Let us first show that $\prob{B}$ can be lower bounded by a constant.  To prove this, we select $2\ell$ vertices with replacement uniformly at random, and pair consecutive vertices to obtain $\ell$ edges.  Choosing vertices uniformly in this way will simplify showing that the degree constraint is satisfied.  Of course we want to avoid ``self-loops'', or edges where both end vertices are the same.  Let $C$ denote the event that there is a vertex that is incident to more than $t$ edges of the sample.  Also, let $D$ denote the event that the sample contains no self-loops and no duplicate edges.  Then

\[
  \mathbb{P}[\bar{B}] = \prob{C|D} \leq \frac{\prob{C}}{\prob{D}}.
\]
Now, the probability of encountering a self-loop is $1/N$ and the probability of an edge being a duplicate of another is at most $2/N^2$. Therefore,
\[
  \mathbb{P}[\bar{D}] \leq \frac{\ell}{N} + \binom{\ell}{2} \frac{2}{N^2} \leq \frac{2\ell}{N}.
\]
By \cref{eq:convergence}, $\ell/N \to 0$. So $\prob{D}$ is bounded away from 0. Moreover, since the edges now consist of pairs of independently chosen vertices, we can approximate the number of edges incident to each vertex by $N$ independent Poisson random variables with parameter $2\ell/N$ thusly:
\[
  \prob{C} \leq N \Round{\frac{e^{-2\ell/N}(2\ell e/N)^t}{t^t}} (e\sqrt{2\ell}),
\]
where the middle factor is a bound on the probability that one Poisson variable is at least $t$ (Theorem 5.4 of~\cite{MitUpf2005}), and the last factor is an adjustment factor for this approximation (Corollary 5.9 of~\cite{MitUpf2005}). This expression is bounded by a constant factor times the expression on the left-hand side of \cref{eq:convergence}. Consequently, $\mathbb{P}[\bar{B}]$ converges to 0, and for sufficiently large $N$, $\mathbb{P}[B] \geq 1/2$, as was to be shown.

Now we find an upper bound for $\prob{A}$. To do this, we select $\ell$ edges with replacement uniformly at random. Because $\ell$ is relatively small when compared to $N$, it is unlikely that the sample will contain any duplicates. Formally, let $E$ be the event that the sample contains at least $R$ red edges, and $F$ be the event that the sample consists of distinct edges. We have
\[
  \prob{A} = \prob{E|F} \leq \frac{\prob{E}}{\prob{F}}.
\]
The probability that two selected edges are the same edge is $1/\binom{N}{2}$. So
\[
  \mathbb{P}[\bar{F}] \leq \binom{\ell}{2} / \binom{N}{2} \leq \frac{\ell^2}{2} \frac{4}{N^2} = 2\Round{\frac{\ell}{N}}^2.
\]
So for large enough $N$, $\prob{F}$ is bounded below by $\frac{1}{2}$. 

Finally, we bound $\prob{E}$. The expected number $X$ of red edges in this sample is $\expec{X} = \ell r / \binom{N}{2} $ which is bounded below by $2\ell r/N^2$ and bounded above by $4\ell r/N^2 = R/2$. 
So using these bounds and a Chernoff bound (\cref{lem:chernoff}),
where we set $\lambda$ equal to $\expec{X}$, we have that
\[
  \prob{E} \leq \exp\Round{-\frac{6}{7} \frac{2\ell r}{N^2}}.
\]
If $\ell r/N^2 \to 0$ as $n \to \infty$, set $\lambda$ equal to $1 - \expec{X} = 1 - \Theta(\ell r/N^2)$:
\[
\prob{E} \leq \exp\Round{-\frac{c}{\expec{X}}}
\]
for some constant $c> 0$. Putting it all together, we have that for large enough $N$, $\prob{B}\geq 1/2$ and $\prob{A}$ is bounded above by $2$ times one of the two bounds for $\prob{E}$. This proves the result.
}

\newcommand{\proofThmPowerDetAdv}{
The proof is similar to the proof of \cref{thm:erdos-det-adv}. An upper bound on the advantage of any deterministic adversary $A: \mathcal{G} \to \mathcal{G}$ on graphs on $n$ vertices is given by the conditional probability
\[
\prob{\identify(z, G, \id_1,\ldots,\id_k, G_A) \neq \id | \dist(G, G_A) < \theta},
\]
where the parameters passed to $\identify$ are defined according to the experiment in \cref{alg:security}. We show that this quantity is polynomially negligible.

For $G_A$ to be successfully identified, it is sufficient for the following three conditions to hold:
\begin{compactenum}[1.]
\item the original graph $G = G(\vector{w}^\gamma)$ is $(4\log n, 4\log n)$-separated;
\item the Hamming distance between any two $\id$ and $\id'$ involved in a pair in $S$ is at least $D = 4(2C'' + C') \log n$;
\item $A$ changes fewer than $D/2$ edges of the watermark.
\end{compactenum}
The proof is similar to the proof of \cref{thm:erdos-det-adv}. To be able to apply \cref{lem:guess}, we need to show that \cref{eq:convergence} holds with $N = \ml$. Recall that
\[
p = K_0 \Round{n^{\gamma - 3}\ml^2}^{-\frac{1}{\gamma - 1}}, \qquad
\ml= K_2\ n^{\Gamma}\Round{\log n}^{\Gamma'}
\]
where
\[
\Gamma = -\frac{2\gamma^2 - 8  \gamma + 5}{2\gamma - 1}, \qquad \Gamma' = -\frac{3(\gamma - 1)^2}{2\gamma - 1}.
\]
By the definition of $\ell$s given in the statement of this theorem, we have
\[
\ell 
= a\frac{1}{K_0} K_2^{2(\frac{1}{\gamma - 1})} n^{-\frac{2\gamma - 7}{2\gamma -1}} \Round{\log n}^{\frac{5-4\gamma}{2\gamma - 1}}
\]
for a constant $a$. Therefore
\begin{equation}
\label{eq:power-convergence}
\frac{\ell^{t+1}}{\ml^{t-1}}
= b (c^{L_1(t)}) (n^{L_2(t)}) (\log n)^{L_3(t)},
\end{equation}
where $b$ and $c$ are constants and $L_1, L_2$ and $L_3$ are linear functions of $t$ that are parameterized by $\gamma$. In particular, 
\[
L_2(t) = \frac{2(\gamma - 2)(\gamma - 3)}{2\gamma -1}t + \frac{2(-\gamma^2+3\gamma+1)}{2\gamma-1}.
\]
For the range of values of $\gamma$ we are concerned with (i.e., $5/2< \gamma < 3$), the first factor is negative. 
Since $t = \log n$ is positive, this shows that \cref{eq:power-convergence} converges to 0 as $n \to \infty$.

Now we prove that each of the three listed conditions fails with polynomially negligible probability. We invoke \cref{lem:power-sep} to show that this is the case for the first condition. For the second condition, we use the fact that each bit is independently set to 1 with some probability $P[i,j]$ where $i\leq \ml$ and $j\leq\ml$. Thus $p = P[\ml,\ml]$ is a lower bound on these probabilities. This, together with the definition of $\ell$ given in this theorem, allow for the hypotheses of \cref{lem:id-sep} to be met. Thus, we can apply this lemma and show that condition 2 fails with polynomially negligible probability. 

Finally, we have shown in our earlier discussion leading up to the\cref{eq:convergence} that the conditions of \cref{lem:guess} are met. Let us identify the edges chosen by the adversary as the red edges, and let the $\ell$-sampled edges be the ones that the watermarking algorithm selected in procedure $\keygen$. By \cref{lem:guess}, the number of edges that are common to both selections is at least
\[
R = \frac{8\ell r}{N^2} = \frac{16Dr}{pN^2} = \frac{D}{2}
\]
where we used the fact that $\ell = 2D/p$, is at most
\[
4 \exp\Round{-\frac{12}{7} \frac{\ell r}{N^2}} = 4 \exp\Round{-\frac{12}{7} \frac{D}{16}} = 4n^{-\frac{3}{7}(2C''+C')}.
\]
So we have our result.
}

\begin{document}

\maketitle

\begin{abstract}
We introduce models and algorithmic foundations
for graph watermarking. 
\ifFull
Our frameworks include security definitions and proofs, 
as well as characterizations when graph watermarking is 
algorithmically feasible, in spite of the fact that the general problem
is NP-complete by simple reductions from the subgraph isomorphism 
or graph edit distance problems.
 In the digital watermarking of many types of files, an implicit step in the recovery of a watermark is the mapping of individual pieces of data, such as image pixels or movie frames, from one object to another. 
 In graphs, this step corresponds to approximately matching vertices of one graph to another based on graph invariants such as vertex degree.
\fi
Our approach is based on characterizing the feasibility 
of graph watermarking
in terms of keygen, marking, and identification functions
defined over graph families with known distributions. We demonstrate 
the strength of this approach with exemplary watermarking schemes
for two random graph models,
the classic Erd\H{o}s-R\'{e}nyi model and a random power-law graph model,
both of which are used to model real-world networks. 
\end{abstract}

\section{Introduction}

In the classic media watermarking problem,
we are given a digital representation, $R$, 
for some media object, $O$, such as a piece of music, a video, or an image, 
such that there is a rich space, $\cal R$,
of possible representations for $O$ besides 
$R$ that are all more-or-less equivalent.
Informally, a \emph{digital watermarking} scheme for $O$ is a function
that maps $R$ and a reasonably short random message, $m$,
to an alternative
representation, $R'$, for $O$ in $\cal R$. 
The verification of such a marking scheme takes $R$ and a 
presumably-marked representation, $R''$ (which was possibly altered by
an adversary), 
along with the set of messages previously used for marking,
and it either
identifies the message from this set that was assigned to $R''$ 
or it indicates a failure.
Ideally, it should difficult for an adversary
to transform a representation, $R'$ (which he was given),
into another representation $R''$ in $\cal R$, 
that causes the identification function to fail.
\ifFull
Some example applications of such digital watermarking schemes include 
steganographic communication and marking digital works for copyright protection
(e.g., see~\cite{cox2007digital,katzenbeisser2000information,shih2007digital}).
\fi

With respect to digital representations of media objects that are intended to be
rendered for human performances, such as music, videos, and images, there is
a well-established literature on digital watermarking schemes and even
well-developed models for such schemes
(e.g., see Hopper {\it et al.}~\cite{hmw-watermarking-07}).
Typically, such watermarking schemes take advantage of the fact that rendered
works have many possible representations with almost imperceptibly
different renderings from the perspective of a human viewer or listener.

In this paper, we are inspired by recent systems work 
on \emph{graph watermarking} by 
\ifFull
Zhao {\it et al.}~\cite{DBLP:journals/corr/ZhaoLZZZ15,Zhao:2015},
\fi
who propose a digital watermarking scheme for graphs,
such as social networks, protein-interaction graphs, etc., which are
to be used for commercial, entertainment, or scientific purposes.
This work by Zhao {\it et al.}~presents a system and experimental results
for their particular method for performing
graph watermarking, but it is lacking in formal security
and algorithmic foundations.
For example, Zhao {\it et al.}~do not provide formal 
proofs for circumstances under which graph watermarking is undetectable
or when it is computationally feasible.
Thus, as complementary work to the 
systems results of Zhao {\it et al.},
we are interested in the present paper in providing models 
and algorithms for graph watermarking, in the spirit of the
watermarking model provided by
Hopper {\it et al.}~\cite{hmw-watermarking-07} 
for media files.
In particular, we are interested in providing a framework
for identifying when graph watermarking is secure and 
computationally feasible.

\subsection{Additional Related Work}
Under the term ``graph watermarking,''
there is some additional work, although it is not actually
for the problem of graph watermarking as we are defining it.
For instance, there is a
line of research involving software watermarking using
graph-theoretic concepts and encodings.
In this case,
the object being marked is a piece of software and the goal of
a ``graph watermarking''
scheme is to create a graph, $G$, from a message, $m$,
and
then embed $G$ into the control flow of a piece of software, $S$, to mark
$S$.
Examples of such work include pioneering work by
Collberg and Thomborson~\cite{Collberg:1999:SWM:292540.292569},
as well as subsequent work by Venkatesan, Vazirani, and Sinha~\cite{vvs-ih-01}
and Collberg {\it et al.}~\cite{ckct-gt-03}.
\ifFull
(See also Chen {\it et al.}~\cite{c5283190} and 
Bento {\it et al.}~\cite{bbm-tpe-13},
as well as a survey by Hamilton and Danicic~\cite{hamilton2010survey}.)
\fi
This work on software watermarking differs from the graph watermarking
problem we study in the present paper, however, because
in the graph watermarking problem we study an input graph 
is provided and we want to alter it to add a mark.
In the graph-based software watermarking problem, 
a graph is instead created from a
message to have a specific, known structure, such as being a permutation
graph, and then that graph is embedded into the control flow of the piece of
software.

A line of research that is more related to the graph watermarking
problem we study is anonymization and de-anonymization for social
\ifFull
networks (e.g., see~\cite{Backstrom:2011:WAT:2043174.2043199,b4497459,%
Hay:2008:RSR:1453856.1453873,k6320456,Liu:2008:TIA:1376616.1376629,%
n5207644,w5504716}).
\fi
One of the closest examples of such prior work is
by Backstrom, Dwork, and Kleinberg~\cite{Backstrom:2011:WAT:2043174.2043199},
who show how to introduce a small set of ``rogue'' vertices into a social
network and connect them to each other and to 
other vertices so that if that same network
is approximately replicated in another setting it is easy to
match the two copies.
Such work differs from graph watermarking, however, because the set of rogue
vertices are designed to ``stand out'' from the rest of the graph rather than
``blend in,'' and it may in some cases be relatively easy for an adversary to
identify and remove such rogue vertices.
\ifFull
Also, we would ideally prefer graph watermarking schemes that make small
changes to the adjacencies of existing vertices rather than mark a graph by
introducing new vertices, since in some applications
it may not be possible to introduce new vertices into a graph that we wish
to watermark.
\fi
In addition to this work,
also of note is work by Narayanan and Shmatikov~\cite{n5207644},
who study the problem of approximately matching two social networks without
marking, as well as the work on Khanna and Zane~\cite{Khanna}
for watermarking road networks
by perterbing vertex positions (which is a marking method outside the scope
of our approach).

\ifFull
Our approach to graph watermarking is also necessarily related to the problem of
graph isomorphism and its approximation
(e.g., see~\cite{Babai15,doi:10.1137/0209047,Czajka200885,%
EPFL-ARTICLE-207759,kobler2012graph,JGT:JGT3190010410}).
In the graph isomorphism problem, we are given two $n$-vertex
graphs, $G$ and $H$, and 
asked if there is a mapping, $\chi$, 
of vertices in $G$ to vertices in $H$ such that
$(v,w)$ is an edge in $G$ if and only if $(\chi(v),\chi(w))$ is an edge in
$H$.
While the graph isomorphism problem is ``famous'' for having 
an uncertain, but unlikely~\cite{Babai15}, 
with respect to being NP-complete,
extensions to subgraph isomorphism and graph edit distance
are known to be NP-complete (e.g., see~\cite{garey}).

There is, of course, also prior work on digital watermarking in general.
For background on such work, we refer the interested reader
to any of the existing surveys, framework papers, or books
(e.g., see~\cite{cox2007digital,hmw-watermarking-07,%
katzenbeisser2000information,shih2007digital}).
\fi

\subsection{Our Results}
In this paper, we introduce a general 
graph watermarking framework that is based on the use of key generation, marking,
and identification functions, as well as a hypothetical watermarking security experiment
(which would be performed by an adversary).
We define these functions in terms of graphs taken over random families of 
graphs, which allows us to quantify situations in which graph watermarking
is provably feasible.

We also provide some graph watermarking schemes as examples of our framework, defined
in terms of the classic Erd\H{o}s-R\'{e}nyi random-graph model and a
random power-law graph model.
Our schemes extend and build upon 
previous results on graph
isomorphism for these graph families, which may be of independent interest.
In particular, we design simple marking schemes for these random graph families
based on simple edge-flipping
strategies involving high- and medium-degree vertices.
Analyzing the correctness of our schemes is quite nontrivial, however, and 
our analysis and proofs involve intricate probabilistic 
\ifFull
arguments.
\fi
We provide an analysis of our scheme against
adversaries that can themselves flip edges in order to defeat our 
mark identification algorithms.
In addition, we provide experimental validation
of our algorithms, showing that our edge-flipping scheme can succeed
for a graph
without specific knowledge of the parameters of its deriving graph family.
\ifFull
We also conducted experiments to fit real-world networks to the random 
power-law graph model, which gave results that showed that the model 
was generally a good fit for the networks tested but the learned values 
did not fall into the range needed for our scheme.
\fi


\section{Our Watermarking Framework}
\label{sec:framework}

\ifFull
We begin by presenting a general framework for graph watermarking, which
differs from the general model
of Hopper {\it et al.}~\cite{hmw-watermarking-07}, but is
similar in spirit.
\fi

Suppose we are given an undirected graph, $G=(V,E)$, that we wish to
mark.
To define the security of a watermarking scheme for $G$, 
$G$ must come from a family of graphs with some degree
of entropy~\cite{DBLP:journals/corr/ZhaoLZZZ15}.
We formalize this by assuming a probability distribution
$\mathcal{D}$ over the family $\mathcal{G}$ of graphs from which $G$ 
is taken.

\begin{definition}
A \emph{graph watermarking scheme} is a 
tuple 
$(\keygen, \wmark, \identify)$ over a set, $\mathcal{G}$,
of graphs where
\begin{compactitem}
\item
$\keygen: \N\times \N \to \Aux$ is a private key generation function,
such that $\keygen(\ell,n)$ is a list of $\ell$ (pseudo-)random graph elements,
such as vertices and/or vertex pairs, defined over a graph of $n$ vertices.
These candidate locations for marking are defined independent
of a specific graph; that is, vertices in $\Aux$ are identified 
simply by the numbering from $1$ to $n$.
For example, $\keygen(\ell,n)$ could be a small random graph, $H$,
and some random edges to connect $H$ to a larger
input graph~\cite{DBLP:journals/corr/ZhaoLZZZ15},
or $\keygen(\ell,n)$ could be a set of vertex pairs in an
input graph that form candidate locations for marking.
\item
  $\wmark: \Aux \times \mathcal{G} \to \N \times \mathcal{G}$ 
takes a private key $z$ generated by $\keygen$, and a specific graph $G$
from $\mathcal{G}$, 
and returns a pair, $S = (\id, H)$,
such that $\id$ is a unique identifier for $H$ and $H$ is the graph obtained by adding the mark determined by $\id$ to $G$ in the location determined determined by the private key $z$. $\wmark$ is called every time a different marked copy needs to be produced, with the $i$-th copy being denoted by $S_i = (\id_i, H_i)$. Therefore, the unique identifiers should be thought of as being generated randomly. To associate a marked graph $H_i$ with the user who receives it, the watermarking scheme can be augmented with a table storing user name and unique identifiers. Alternatively, the identifiers can be generated pseudo-randomly as a hash of a private key provided by the user.
\item
  $\identify: \Aux \times \mathcal{G} \times \N^k \times \mathcal{G} \to \N \cup \set{\perp}$ 
takes a private key from $\Aux$, 
the original graph, $G$, 
$k$ identifiers of previously-marked copies of 
$G$, and a test graph, $G'$, and it returns
 the identifier, $\id_i$, of the watermarked graph
that it is identifying as a match for $G'$. It may also return $\perp$,
as an indication of failure,
if it does not identify any of the graphs $H_i$ as a match for $G'$.
\end{compactitem}
In addition, in order for a watermarking scheme to be effective, we require 
that \emph{with high probability}\footnote{Or ``\emph{whp},'' 
that is, with probability at least $1 - O(n^{-a})$, for some $a>0$.}
over the graphs from $\mathcal{G}$ and $k$ output pairs, $S_1,\ldots,S_k$ 
of $\wmark(z, G)$, 
for any $(\id, G')=S_i$, we have $\identify(z, G, \id_1,\ldots,\id_k, G') = \id$.
\end{definition}

\cref{alg:security} shows a hypothetical 
security experiment for a watermarking scheme with respect
to an adversary, $A : \mathcal{G} \to \mathcal{G}$, who is trying to
defeat the scheme.
Intuitively, in the hypothetical experiment, we generate a key $z$, choose
a graph $G$, from family $\mathcal{G}$ according to distribution
$\mathcal{D}$ (as discussed above), 
and then generate $k$ marked graphs according to our scheme (for some set
of $k$ messages).
Next, we randomly choose one of the marked graphs, $G'$, and communicate 
it to an adversary.
The adversary then outputs a graph $G_A$ that is similar to $G'$
where his goal is to cause our identification algorithm to fail on~$G_A$.

\begin{algorithm}[h]
$\mathsf{experiment}(A, k, \ell, n)$:
\begin{compactenum}
\item $z \gets \keygen(\ell,n)$
\item $G \gets_{\mathcal{D}} \mathcal{G}$
\item $S_i \gets \mathsf{mark}(z, G)$, for $i=1,\ldots,k$
\item randomly choose $S_i=(\id, G')$ from $\{S_1,\ldots,S_k\}$
\item $G_A \gets A(G')$
\end{compactenum}
\caption{Hypothetical Watermarking Security Experiment}
\label{alg:security}
\end{algorithm}

In order to characterize differences between graphs, we assume
a similarity measure $\dist: \mathcal{G} \times \mathcal{G} \to \R$, 
defining the distance between graphs in family $\mathcal{G}$.
We also include
a similarity threshold $\theta$, 
that defines
the advantage of an adversary performing
the experiment in \cref{alg:security}.
Specifically,
the \emph{advantage} of an adversary, $A: \mathcal{G} \to \mathcal{G}$
who  is trying to defeat our watermarking scheme is
\[
  \prob{\dist(G, G_A) < \theta \text{ and } \identify(z, G, \id_1,\ldots,\id_k, G_A) \neq \id}.
\]
The watermarking scheme is $(\mathcal{D}, \dist, \theta, k, \ell)$-secure against adversary $A$ if the similarity threshold is $\theta$ and $A$'s advantage 
is \emph{polynomially negligible} (i.e., is $O(n^{-a})$ for some $a > 0$).

Examples of adversaries could include the following:
\begin{compactitem}
\item
\emph{Arbitrary edge-flipping adversary}: a malicious adversary
who can arbitrarily flip edges in the graph. That is,
the adversary adds an edge if it is not already there, and removes it otherwise.
\item
\emph{Random edge-flipping adversary}: an adversary who independently flips 
each edge with a given probability.
\item
\emph{Arbitrary adversary}:
a malicious adversary
who can arbitrarily add and/or remove vertices and flip edges in the graph.
\item
\emph{Random adversary}: an adversary who 
independently adds and/or removes vertices with a given probability and
independently flips each edge with a given probability.
\end{compactitem}
\ifFull
One could also imagine other types of adversaries, as well, such as a random
adversary who is limited in terms of the numbers or types of edges or vertices
that he can change.
\fi

\subsection{Random graph models}
As defined above, a graph watermarking scheme requires that graphs to be marked
come from some distribution.  In this paper, we consider two 
families of random graphs---the classic
Erd\H{o}s-R\'{e}nyi model and a random power-law graph model---which should
capture large classes of applications where graph watermarking would be of
interest.

\begin{definition}[The Erd\H{o}s-R\'{e}nyi model]
A \emph{random graph} $G(n,p)$ is a graph with $n$ vertices, where each of the $\binom{n}{2}$ possible edges appears in the graph independently with probability $p$.
\end{definition}

\begin{definition}[The random power-law graph model, \S5.3 of \cite{ChungLu2006}]
\label{def:rplg}
Given a sequence $\vector{w} = (w_1, w_2, \dots, w_n)$, such that $\max_i w_i^2 < \sum_k w_k$, the \emph{general random graph} $G(\vector{w})$ is defined by labeling the vertices $1$ through $n$ and choosing each edge $(i, j)$ independently from the others with probability $p[i,j] = \rho w_i w_j$, where $\rho = 1/ \sum_j w_j$.

We define a \emph{random power-law graph} $G(\vector{w}^\gamma)$ parameterized by the maximum degree $m$ and average degree $w$. Let $w_i = ci^{-1/(\gamma - 1)}$ for values of $i$ in the range between $i_0$ and $i_0 + n$, where
\begin{equation}
\label{eq:rplg}
c = \frac{\gamma - 2}{\gamma -1} wn^{\frac{1}{\gamma - 1}}, \qquad
i_0 = n \Round{\frac{w(\gamma - 2)}{m(\gamma - 1)}}^{\gamma - 1}.
\end{equation}
This definition implies that each edge $(i,j)$ appears with probability
\begin{equation}
\label{eq:powerlaw-probs}
P[i,j] = K_0 \Round{n^{\gamma - 3}ij}^{-\frac{1}{\gamma - 1}}, \quad \text{where } K_0 \defeq \Round{\frac{\gamma - 2}{\gamma - 1}}^2 w.
\end{equation}
\end{definition}

\ifFull
As we show in the following proposition, this model does indeed have
a power-law degree distribution.

\begin{proposition}
\label{prop:expected-degree}
In the random power-law graph $G(\vector{w}^\gamma)$, the expected number of vertices with degree $k$ is between $Cn/k^{\gamma}$ and $Cn/(k+1)^{\gamma}$ where $C = \Round{w(\gamma - 2)}^{\gamma - 1}/(\gamma - 1)^{\gamma - 2}$.
\end{proposition}

\begin{proof}
The function $i(k) = (c/k)^{\gamma - 1}$ relating the index of a vertex to its expected degree $k=w_i$ is convex and decreasing. By the mean value theorem, the number $\Delta i$ of indices $j$ such that $k \leq w_j < k+1$ satisfies
\[
|i'(k)| \leq \frac{\Delta i}{(k+1) - k} = \Delta i \leq |i'(k+1)|.
\]
Now the derivative of $i(k)$ is $-Cn/k^{\gamma}$. Noting that $\Delta i$ is the expected number of vertices of degree $k$, the result is proven.
\end{proof}
\fi

\subsection{Graph watermarking algorithms}
We discuss some instantiations of the graph watermarking framework
defined above. Unlike previous watermarking or de-anonymization
schemes that add vertices~\cite{Backstrom:2011:WAT:2043174.2043199,DBLP:journals/corr/ZhaoLZZZ15}, we describe an effective and
efficient scheme based solely on edge flipping.  Such an approach
would be especially useful for applications where it could be infeasible
to add vertices as part of a watermark.

Our scheme does not require adding labels to the vertices or
additional objects stored in the graph for
identification purposes. Instead, we simply rely on the structural
properties of graphs for the purposes of marking.  In particular, we focus on
the use of vertex degrees, that is, the number of edges incident on each
vertex. We identify high and medium degree vertices as candidates for finding 
edges that can be flipped in the course of marking. The specific 
degree thresholds for what we mean by
``high-degree'' and ``medium-degree'' depend on the graph family, however, 
so we postpone defining these notions precisely until our analysis sections.

Algorithms providing an example implementation of our graph watermarking
scheme are shown in \cref{alg:scheme}.
The $\keygen$ algorithm randomly selects a set of 
candidate vertex pairs for flipping, from among the high- and 
medium-degree vertices,
with no vertex being incident to more than a parameter $t$ of candidate
pairs.
We introduce a 
procedure, $\wlabel(G)$,
which labels high-degree vertices by their degree ranks and 
each medium-degree vertex, $w$,
by a bit vector identifying its high-degree adjacencies.
This bit vector has a bit for each high-degree vertex, which is $1$ for neighbors of $w$ and $0$ for non-neighbors.
The algorithm
$\wmark(z, G)$,
takes a random set of candidate edges and a graph, $G$,
and it flips the corresponding edges in $G$ according to a resampling
of the edges using the distribution $\mathcal{D}$.
The algorithm,
\textsf{approximate-isomorphism}$(G,H)$,
returns a mapping of the high- and medium-degree vertices in $G$
to matching high- and medium-degree vertices in $H$, if possible.
The algorithm,
$\identify(z, G, \id_1, \ldots, \id_k, H)$,
uses the approximate isomorphism algorithm to match up high- and medium-degree
vertices in $G$ and $H$, and then it extracts the bit-vector from this matching
using $z$.

\begin{algorithm}
$t$: the maximum number of flipped edges that can be adjacent to the same vertex.
$\keygen(\ell, n)$:\\
\begin{compactenum}
\item Let $x$ denote the total number of high- and medium-degree vertices
\item $X = \varset{(u,v)}{1\leq u < v \leq x}$
\item Let $z$ be a list of $\ell$ pairs randomly sampled (without replacement) 
from $X$ such that no end vertex appears more than $t$ times
\item return $z$
\end{compactenum}
$\wlabel(G)$:
\begin{compactenum}
\item sort the vertices in decreasing order by degree and identify the high- and medium-degree vertices
\item if the degrees of high-degree vertices are not unique, return failure
\item label each high-degree vertex with its position in the vertex sequence
\item label each medium-degree vertex with a bit vector encoding its high-degree adjacencies
\item if the bit vectors are not unique, return failure
\item otherwise, return the labelings
\end{compactenum}
$\wmark(z, G)$:\\
\begin{compactenum}
\item $S = \emptyset$
\item $V$ is the set of high- and medium-degree vertices of $G$, sorted lexicographically by their labels given by $L = \wlabel(G)$
\item generate an $\ell$-bit string $\id$ where each bit $i$
is independently set to 1 
with probability $p_{z[i]}$, where $p_{z[i]}$ is the probability of 
the edge $z[i]$ in $\mathcal{D}$
\item let $H$ be a copy of $G$
\item for $j$ from 1 to $\ell$:
\item \quad $(u,v) = z[j]$
\item \quad if $\id[j]$ is 1:
\item \quad \quad insert edge $(V[u], V[v])$ in $H$
\item \quad else:
\item \quad \quad remove edge $(V[u], V[v])$ from $H$ 
\item return $(\id, H)$ 
\end{compactenum}
\textsf{approximate-isomorphism}$(G,H)$:
\begin{compactenum}
\item call $\mathsf{label}(G)$ and $\mathsf{label}(H)$, returning failure if either of these fail.
\item match each of $G$'s high-degree vertices with the vertex in $H$ with the same label.
\item match each of $G$'s medium-degree vertices with the vertex in $H$ whose label is closest in Hamming distance.
\item if $H$ has a vertex that is matched more than once, return failure.
\item otherwise, return the (partial) vertex assignments between $G$ and $H$.
\end{compactenum}
$\identify(z, G, \id_1, \ldots, \id_k, H)$:
\begin{compactenum}
\item
find an \textsf{approximate-isomorphism}($G,H$), returning $\perp$ if failure occurred at any step.
\item $V$ is the set of high- and medium-degree vertices of $G$, sorted lexicographically by their labels given by $L = \wlabel(G)$
\item $V'$ is the set of vertices of $H$ identified as corresponding to those in $V$, in that same order.
\item $\id$ is an empty bit string
\item for $(u,v)$ in $z$ (from left to right):
\item \quad $b = 1$ iff there is an edge between $V'[u]$ and $V'[v]$ in $H$. 
\item \quad append $b$ to $\id$
\item return among the $\id_i$'s the one closest to $\id$
\end{compactenum}
\caption{Watermarking scheme for random graphs.}
\label{alg:scheme}
\end{algorithm}

As mentioned above, we also need a notion of distance for graphs. We use two
different such notions.
The first is the graph edit distance, which is the minimum number of edges 
needed to flip to go from one graph to another.
The second is vertex distance, which intuitively is an edge-flipping
metric localized to vertices.

\begin{definition}[Graph distances]
Let $\mathcal{G}$ be the set of graphs on $n$ vertices. If $G, H \in \mathcal{G}$, define $\Pi$ as the set of bijections between the vertex sets $V(G)$ and $V(H)$. Define the \emph{graph edit distance} $\dist_e: \mathcal{G} \times \mathcal{G} \to \N$ as 
\[
\dist_e(G, H) = \min_{\pi \in \Pi} \abs{E(G) \oplus_\pi E(H)},
\]
where $\oplus_\pi$ is the symmetric difference of the two edge sets under correspondence $\pi$. Define the \emph{vertex distance} $\dist_v: \mathcal{G} \times \mathcal{G} \to \N$ as
\[
\dist_v(G, H) = \min_{\pi \in \Pi} \max_{v \in V(G)} \abs{E(v) \oplus_\pi E(\pi(v))},
\]
where $E(v)$ is the set of edges incident to $v$.
\end{definition}

\section{Identifying High- and Medium-Degree Vertices}
We begin analyzing our proposed graph watermarking scheme by showing
how high- and medium-degree vertices can be identified under our two
random graph distributions. 
\ifFull
We begin with some technical results related to graph isomorphism that form the basis of our watermarking approach, with
the goal of determining the conditions under 
which a vertex of a random graph can be identified with high probability, 
either by its 
degree (if the degree is high) or by its set of high-degree 
neighbors (if it has medium degree).
\fi
We ignore low-degree vertices: their information content 
and distinguishability are low, and they are not used by our example scheme.
\ifFull
Because our results on vertex identifiability are used in our graph watermarking scheme, 
we also determine how robust these identifications are, 
based on how well-separated the vertices are by their degrees. 
\fi

We first find a threshold number $k$ such that the $k$ vertices with highest degree are likely to have distinct and well-separated degree values. We call these $
k$ vertices the \emph{high-degree} vertices.
Next, we look among the remaining vertices for those that are well-separated in 
terms of their high-degree neighbors. 
Specifically, the (high-degree) \emph{neighborhood distance} between two 
vertices is the number of high-degree vertices which are connected to 
exactly one of the two vertices. 
Note that we will omit the term ``high-degree'' in 
``high-degree neighborhood distance'' from now on, as it will always be implied.

In the Erd\H{o}s-R\'{e}nyi model, we show that all vertices that are not 
high-degree nevertheless have well-separated high-degree neighborhoods whp. 
In the random power-law graph model, however, there will be many 
lower-degree vertices whose high-degree neighborhoods cannot be separated. 
Those that have well-separated high-degree neighborhoods with high probability form the
medium-degree vertices, and the rest are the low-degree vertices.

For completeness, we include the following well-known 
Chernoff concentration bound, which we will refer to time and again. 

\begin{lemma}[Chernoff inequality~\cite{ChungLu2006}]
\label{lem:chernoff}
  Let $X_1, \dots, X_n$ be independent random variables with 
  \[
    \prob{X_i = 1} = p_i, \qquad \prob{X_i = 0} = 1 - p_i.
  \]
  We consider the sum $X = \sum_{i=1}^n X_i$, with expectation $\expec{X} = \sum_{i = 1}^n p_i$. Then
  \begin{align*}
    \prob{X \leq \expec{X} - \lambda} &\leq e^{-\frac{\lambda^2}{2\expec{X}}}, \\
    \prob{X \geq \expec{X} + \lambda} &\leq e^{-\frac{\lambda^2}{2\expec{X}+\lambda/3}}.
  \end{align*}
\end{lemma}

\subsection{Vertex separation in the Erd\H{o}s-R\'{e}nyi model}
Let us next consider vertex separation results for the
classic Erd\H{o}s-R\'{e}nyi random-graph model.
Recall that in this model, each edge is chosen independently with probability
$p$.

\begin{definition}
Index vertices in non-increasing order by degree. Let $d_i$ represent the $i$-th highest degree in the graph. Given $h= O(n)$, we say that a vertex is \emph{high-degree} with respect to $d_h$ if it has degree at least $d_h$. Otherwise, we say that the vertex is \emph{medium-degree}. 
\ifFull
We just say high-degree when the value of $h$ is understood from context.
\fi
\end{definition}

Note that in this random-graph model, there are no low-degree vertices.

\begin{definition}
A graph is \emph{$(d,d')$-separated} if all high-degree vertices differ in their degree by at least $d$ and all medium-degree vertices are neighborhood distance $d'$ apart.\end{definition}

Note: this definition depends on how high-degree or medium-degree vertices are defined and will therefore be different for the random power-law graph model.

\begin{lemma}[Extension of Theorem 3.15 in \cite{Bollobas1998}]
\label{lem:erdos-deg-sep}
Suppose $m = o(pqn/\log n)^{1/4}$, $m \to \infty$, and $\alpha(n) \to 0$. Then with probability
\[
1 - m \alpha(n) - 1/\Square{m \Round{\log (n/m)}^2},
\]
$G(n,p)$ is such that
\[
d_i - d_{i+1} \geq \frac{\alpha(n)}{m^2} \Round{\frac{pqn}{\log n}}^{1/2} \text{ for every } i < m,
\]
where $q = 1 - p$.
\end{lemma}

\begin{proof}
\ifFull
\proofErdosDegSep
\fi
\end{proof}

\begin{lemma}[Vertex separation in the Erd\H{o}s-R\'{e}nyi model]
\label{lem:erdos-sep}
Let $0 < \epsilon < 1/9$, $d \geq 3$, $C \geq 3$, $h = n^{(1-\epsilon)/8}$. Suppose $0 < p = p(n) \leq \frac{1}{2}$ is such that $p = \omega(n^{-\epsilon}\log n)$. Then $G(n,p)$ is $(d, C\log n)$-separated with probability $1 - O(n^{-(1-\epsilon)/8})$.
\end{lemma}

\begin{proof}
\ifFull
\proofLemErdosSep
\fi
\end{proof}

Thus, high-degree vertices are well-separated with high probability in the
Erd\H{o}s-R\'{e}nyi model, and the medium-degree vertices are distinguished
with high probability by their high-degree neighborhoods.

\subsection{Vertex separation in the random power-law graph model}
We next study vertex separation for a random power-law graph model,
which can match the degree distributions of many graphs that 
naturally occur in social networking and science.
For more information about power-law graphs and their applications,
see e.g.~\cite{Caldarelli2013,doi:10.1080/15427951.2004.10129088,Newman2006}.

In the random power-law graph model, vertex indices are used to define edge weights and therefore do not necessarily start at 1. 
The lowest index that corresponds to an actual vertex is denoted $i_0$. 
So vertex indices range from $i_0$ to $i_0 + n$. 
Additionally, there are two other special indices $\hm$ and $\ml$,
which we define in this section,
that separate the three classes of vertices.

\begin{definition}
The vertices ranging from $i_0$ to $\hm$ are the \emph{high-degree} vertices, those that range from $\hm + 1$ to $\ml$ are the \emph{medium-degree} vertices, and those beyond $\ml$ are the \emph{low-degree} vertices.
\end{definition}

In this model,
the value of $i_0$ is constrained by the requirement that $P[i_0,i_0] < 1$. When $\gamma \geq 3$, this constraint is not actually restrictive. However, when $\gamma < 3$, $i_0$ must be asymptotically greater than $n^{-(\gamma - 3)/2}$. The constraints on $i_0$ also constrain the value of the maximal and average degree of the graph.

We define $\hm$ and $\ml$ to be independent of $i_0$, but dependent on parameters that control the amount and probability of separation at each level. 
The constraints that $i_0 < \hm$ and $\hm < \ml$ translate into corresponding restrictions on the valid values of $\gamma$, namely that $\gamma > 5/2$ and $\gamma < 3$.
We define $\hm$ in the following lemma.

\begin{lemma}[Separation of high-degree vertices]
\label{lem:sep-high}
In the $G(\vector{w}^\gamma)$ model, let $\delta_i = \abs{w_{i+1} - w_{i}}/2$. Then,
\begin{equation}
\label{eq:delta}
\frac{c}{2(\gamma - 1)}(i+1)^{-\frac{\gamma}{\gamma - 1}} \leq \delta_i \leq \frac{c}{2(\gamma - 1)}i^{-\frac{\gamma}{\gamma - 1}}.
\end{equation}
Moreover, for all $\epsilon_1$ satisfying $0 < \epsilon_1 \leq 1$ and $C_1 > 0$, the probability that
\[
\abs{\deg(i) - w_i} < \epsilon_1\delta_i \qquad \text{for all } i \leq \hm \defeq \Round{\frac{c\epsilon_1^2}{16(\gamma - 1)^2 C_1 \log n}}^{\frac{\gamma- 1}{2\gamma- 1}} 
\]
is at least $1 - n^{-C_1}$.
\end{lemma}

\begin{proof}
The first statement follows from the fact that $w_i$ is a convex function of $i$ and from taking its derivative at $i$ and $i+1$.

For the second statement, let $C > 0$ and let $\hm' \defeq \Round{\frac{c\epsilon_1^2}{8(\gamma - 1)^2 C \log n}}^{\frac{\gamma - 1}{2\gamma - 1}}$. We will show that if $i \leq \hm'$, then
\begin{equation}\label{eq:sep-high}
\prob{\abs{\deg(i) - w_i}\geq \epsilon_1 \delta_i} < n^{-C}.
\end{equation}
Now we choose $C$ such that $C_1 + \log\hm/\log n < C \leq 2C_1$. The inequality $C \leq 2C_1$ implies that $\hm \leq \hm'$ and (\ref{eq:sep-high}) holds for all $i \leq \hm$. By the union bound applied to \cref{eq:sep-high}
\[
\prob{\exists i \leq \hm, \abs{\deg(i) - w_i}\geq \epsilon_1 \delta_i} \leq \hm n^{-C}.
\]
Since $C_1 + \log\hm/\log n < C$, the right hand side is bounded above by $n^{-C_1}$. This proves the result.

Now, we prove \cref{eq:sep-high}. Clearly, since $\delta_i = (w_{i} - w_{i+1})/2$, we have that $w_i \geq \delta_i$. So if $\epsilon_1 \leq 1$ and $\lambda_i = \epsilon_1 \delta_i$, then $w_i \geq \lambda_i /3$. This implies that
\[
\frac{\lambda_i^2}{w_i + \lambda_i/3} \geq \frac{\lambda_i^2}{2w_i} \geq \frac{c\epsilon_1^2}{8(\gamma - 1)^2} i^{-\frac{2\gamma-1}{\gamma-1}},
\]
where the second inequality follows from \cref{eq:delta} and the definition of $w_i$ given in \cref{def:rplg}. If $i \leq \hm'$, the right hand side is lower-bounded by $C\log n$. The result follows by applying a Chernoff bound (\cref{lem:chernoff}). 
\end{proof}

For simplicity, we often use the following observation.

\begin{observation}
Rewriting $\hm$ to show its dependence on $n$, we have
\begin{equation}
\label{eq:hm}
\hm(\epsilon_1, C_1) = K_1(\epsilon_1, C_1)\ n^{\frac{1}{2\gamma - 1}}\Round{\log n}^{-\frac{\gamma - 1}{2\gamma - 1}}, 
\quad K_1(\epsilon_1, C_1) \defeq \Round{\frac{\gamma - 2}{(\gamma - 1)^3} \frac{w\epsilon_1^2}{16 C_1}}^{\frac{\gamma - 1}{2\gamma - 1}}.
\end{equation}
For the graph model to make sense, the high-degree threshold must be asymptotically greater than the lowest index. In other words, we must have that $i_0 = o(\hm)$. Since $i_0 = \Omega(n^{-(\gamma-3)/2})$, this implies that $\gamma > 5/2$.
\end{observation}

We next define $\ml$, the degree threshold for medium-degree
vertices, in the following lemma.

\begin{lemma}[Separation of medium-degree vertices]
\label{lem:sep-mid}
Let $K_0$ be defined as in \cref{def:rplg}, $K_1(\epsilon_1, C_1)$ be defined as in \cref{eq:hm}, and
\begin{equation}
 K_2(\epsilon_1, C_1, \epsilon_2, C_2) \defeq \frac{K_0^{\gamma -1}K_1^{\gamma - 2}(\epsilon_1, C_1)}{(C_2 + 2\Gamma + 2\log (K_0^{\gamma -1}K_1^{\gamma - 2}(\epsilon_1, C_1)) + 2\epsilon_2)^{\gamma - 1}}.
\end{equation}
Let $X_{ij}$ denote the neighborhood distance between two vertices $i$ and $j$ in $G(\vector{w}^\gamma)$. If $5/2 < \gamma < 3$, for every $\epsilon_2 > 0$ and $C_2 > 0$, the probability that
\[
X_{ij} > \epsilon_2 \log n, \quad \text{for all } \hm \leq i,j \leq \ml
\]
where
\begin{equation}
\label{eq:ml}
\ml(\epsilon_1, C_1, \epsilon_2, C_2) \defeq K_2(\epsilon_1, C_1, \epsilon_2, C_2) \ n^{\Gamma}\Round{\log n}^{-\frac{3(\gamma - 1)^2}{2\gamma - 1}},
\quad \Gamma \defeq -\frac{2\gamma^2 - 8  \gamma + 5}{2\gamma - 1},
\end{equation}
is at least $1 - n^{-C_2}$ for sufficiently large $n$.
\end{lemma}

\begin{proof}
Let $C > 0$ and let
\[
\ml' \defeq \Round{\frac{C_2 + 2\Gamma + 2\log (K_0^{\gamma -1}K_1^{\gamma - 2}) + 2\epsilon_2}{C + 2\epsilon_2}}^{\gamma - 1} \ml.
\]
We claim that if $\hm \leq i,j \leq \ml'$, then
\begin{equation}\label{eq:sep-mid}
\prob{X_{ij} \leq \epsilon_2\log n} \leq n^{-C}.
\end{equation}
If we choose $C = C_2 + 2\Gamma + 2\log K_0^{\gamma -1}K_1^{\gamma - 2}$, we have that $\ml = \ml'$, so that \cref{eq:sep-mid} applies to all $i,j$ such that $i, j \leq \ml$. Moreover, since
\[
\ml \leq K_0^{\gamma - 1} K_1^{\gamma - 2}n^{\Gamma} \leq n^{\log(K_0^{\gamma - 1} K_1^{\gamma - 2})}n^{\Gamma},
\]
our choice of $C$ implies that $\ml^2 \ n^{-C} \leq n^{-C_2}$. By applying the union bound to \cref{eq:sep-mid}, we have
\[
\prob{\exists i, j \text{ s.t. } \hm \leq i,j \leq \ml, \ X_{ij} \leq\ \epsilon_2\log n}\ \leq\ \ml^2 n^{-C}\ \leq\ n^{-C_2},
\]
which establishes the lemma.

Let us now prove the claim. Observe that $X_{ij}$ is the sum over the high-degree vertices $k$, of indicator variables $X^k_{ij}$ for the event that vertex $k$ is connected to exactly one of the vertices $i$ and $j$. It i For fixed $i$ and $j$, these are independent random variables. Therefore, we can apply a Chernoff bound. The probability that $X^k_{ij} = 1$ is
\[
P[i,k] (1 - P[j,k]) + P[j,k] (1 - P[i,k]) \geq 2P[\ml,\hm](1 - P[i_0, \hm]).
\]
Since $P[i_0, \hm] \to 0$, for sufficiently large $n$, this expression is bounded below by $P[\ml, \hm]$, and
\[
\expec{X_{ij}} \geq \hm P[\ml,\hm] \geq (C + 2\epsilon_2)\log n,
\]
by \cref{eq:powerlaw-probs}, \cref{eq:hm} and \cref{eq:ml}, as can be shown by a straightforward but lengthy computation. Let $d = \epsilon_2\log n$. This implies that
\[
\frac{(\expec{X_{ij}} - d)^2}{\expec{X_{ij}}} \geq \expec{X_{ij}} - 2d \geq C \log n.
\]
Therefore, applying the Chernoff bound (\cref{lem:chernoff}) to the $X_{ij}^k$ for fixed $i$ and $j$ and all high-degree vertices $k$ proves the claim.
\end{proof}

\begin{observation}
\label{obs:ml}
We would have the undesirable situation that $\ml = o(1)$ whenever $\frac{2\gamma^2 - 8  \gamma + 5}{2\gamma - 1} > 0$, or equivalently when $\gamma > 2 + \sqrt{3/2} > 3$. In fact, in order for $\hm = o(\ml)$, we must have $\gamma < 3$.
\end{observation}

We illustrate the breakpoints for high-, medium-, and low-degree vertices
in \cref{fig:breaks}.

\begin{figure}[hbt]
\centering
\includegraphics[scale=0.8]{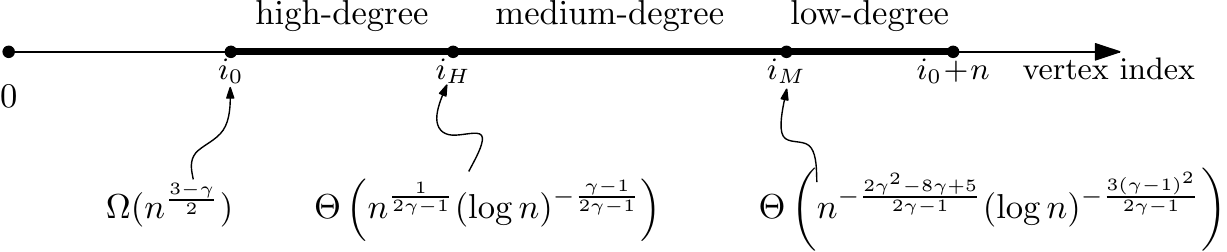}
\caption{Degree breakpoints for the random power-law graph model.}
\label{fig:breaks}
\end{figure}

The next lemma summarizes the above discussion and provides the
forms of $\hm$ and $\ml$ that we use in our analysis.

\begin{lemma}[Vertex separation in the power-law model]
\label{lem:power-sep}
Let $5/2 < \gamma < 3$. Fix $\epsilon > 0, C_1 > 0, C_2 > 0$. Let $\hm = \hm(\epsilon_1, C_1)$ and $\ml = \ml(\epsilon_1, C_1, \epsilon_2, C_2)$ where $\epsilon_1 = 1$ and $\epsilon_2 = \epsilon$. Let 
\[
d = n^{\frac{1}{2\gamma -1}} \quad \text{ and } \quad d' = \log n.
\]
For sufficiently large $n$, the probability that a graph $G(\vector{w}^\gamma)$ is not $(\epsilon d,\epsilon d')$-separated is at most $n^{-C_1} + n^{-C_2}$.
\end{lemma}

\begin{proof}
Let $\delta_i$ be defined as in \cref{lem:sep-high}. A straightforward computation using \cref{eq:rplg}, \cref{eq:delta}, and \cref{eq:hm} shows that
\[
\delta_{\hm} \geq \text{constant} \cdot n^{\frac{1}{2\gamma - 1}} \Round{\log n}^{\frac{\gamma}{2\gamma - 1}}.
\]
So for sufficiently large $n$, we have $\delta_{\hm} \geq 3\epsilon d/2$. For all $i \leq \hm$, the average degrees $w_i$ of consecutive vertices are at least $3\epsilon d/2$ apart. So for two high-degree vertices to be within $\epsilon d$ of each other, at least one of the two must have degree at least $(3\epsilon /2 - \epsilon /2)d$ away from its expected degree. By \cref{lem:sep-high}, the probability that some high-degree vertex $i$ satisfies $\abs{\deg(i) - w_i} > \delta_{\hm}$ is at most $n^{-C_1}$.

By \cref{lem:sep-mid}, the probability that there are two medium-degree vertices with neighborhood distance less than $\epsilon d'$ is at most $n^{-C_2}$.
\end{proof}

Thus, our marking scheme for the random power-law graph model is effective.

\section{Adversary Tolerance}

In this section, we study the degree to which our exemplary graph watermarking
scheme can tolerate an arbitrary edge-flipping adversary. 
\ifFull
To measure success, we use the notion of security and adversary advantage which are formally defined in \ref{sec:framework}. We quantify the number of edge flips that can be tolerated under the Erd\H{o}s-R\'{e}nyi model and the random power-law graph model.
\fi

\begin{theorem}[Security against an arbitrary edge-flipping adversary in the Erd\H{o}s-R\'{e}nyi model]
\label{thm:erdos-det-adv}
Let $0 < \epsilon < 1/9$, $d \geq 3$, $h = n^{(1-\epsilon)/8}$ and $p \leq 1/2$ such that $p = \omega(n^{-\epsilon} \log n)$. Let $d$ be sufficiently large so that 
\begin{equation}
\label{eq:erdos-t}
\epsilon \ \frac{d +1}{d-1} < 1.
\end{equation}
Suppose the similarity measure is the vertex distance $\dist_v$, the similarity threshold is $\theta = d$, we have a number $k = n^C$ of watermarked copies, and their identifiers are generated using $\ell = 8(2C+C')n^\epsilon$ bits. Suppose also that the identifiers map to sets of edges of a graph constrained by the fact that no more than $t = d$ edges can be incident to any vertex. The watermarking scheme defined in \cref{alg:scheme} is $(G(n,p), \dist_v, \theta, k, \ell)$-secure against any deterministic adversary.
\end{theorem}

The proof of this theorem relies on two lemmas.  \Cref{lem:id-sep} identifies conditions under which a set of bit vectors with bits independently set to 1 is unlikely to have two close bit vectors. \Cref{lem:guess} states that a deterministic adversary's ability to guess the location of the watermark is limited. Informally, this is because the watermarked graph was obtained through a random process, so that there are many likely original graphs that could have produced it.

\begin{lemma}[Separation of IDs]
\label{lem:id-sep}
Consider $k=n^C$ random bit strings of length $\ell$, where each bit is independently set to 1, and the i-th bit is 1 with probability $q_i$ satisfying $p \leq q_i \leq 1/2$ for a fixed value $p$. The probability that at least two of these strings are within Hamming distance $D = 4(2C + C') \log n$ of each other is at most $n^{-C'}$ if $\ell p \geq 2D$.
\end{lemma}

\begin{proof}
\ifFull
\proofLemIDSep
\fi
\end{proof}

\begin{lemma}[Guessing power of adversary]
\label{lem:guess}
Consider a complete graph on $N$ vertices, and let $r$ of its edges be red. Let $s$ be a sample of $\ell$ edges chosen uniformly at random among those that satisfy the constraint that no more than $t$ edges of the sample can be incident to any one vertex. Suppose also that $\ell, N$ and $t$ are non-decreasing functions of $n$ such that
\begin{equation}
  \label{eq:convergence}
  \frac{\ell^{t + 1}}{N^{t - 1}} \to 0  \ \text{as} \ n \to \infty. 
\end{equation}
For sufficiently large $N$, the probability that $s$ contains at least $R = 8\ell r/N^2$ red edges is bounded by $4 \exp\Round{-{12\ell r}/(7{N^2})}$. Moreover, if $\ell r/N^2 \to 0$, then the probability that $s$ contains at least $R = 1$ red edge is bounded by $4\exp\Round{-{cN^2}/{(\ell r)}}$, for some $c > 0$ and for sufficiently large $N$.
\end{lemma}

\begin{proof}
\ifFull
\proofLemGuess
\fi
\end{proof}

\begin{proof}[\cref{thm:erdos-det-adv}]
An upper bound on the advantage of any deterministic adversary $A: \mathcal{G} \to \mathcal{G}$ on graphs on $n$ vertices is given by the conditional probability
\[
\prob{\identify(z, G, \id_1,\ldots,\id_k, G_A) \neq \id | \dist_v(G, G_A) < \theta},
\]
where the parameters passed to $\identify$ are defined according to the experiment in \cref{alg:security}. We show that this quantity is polynomially negligible.

For $G_A$ to be successfully identified, it is sufficient for the following three conditions to hold:
\begin{compactenum}[1.]
\item the original graph $G = G(n,p)$ is $(4d, 4d)$-separated;
\item the Hamming distance between any two $\id$ and $\id'$ involved in a pair in $S$ is at least $D = 4(2C + C') \log n$;
\item $A$ changes no edges of the watermark.
\end{compactenum}
These are sufficient conditions because we only test graphs whose vertices had at most $d$ incident edges modified by the adversary, and another $d$ incident edges modified by the watermarking. So for original graphs that are $(4d, 4d)$-separated, the labeling of the vertices can be successfully recovered. Finally, if the adversary does not modify any potential edge that is part of the watermark, the $\id$ of the graph is intact and can be recovered from the labeling.

Now, by \cref{lem:erdos-sep}, the probability that $G(n,p)$ is not $(4d,4d)$-separated is less than $O(n^{-(1-\epsilon)/8})$. Moreover, since $\ell p \geq 2D$, by \cref{lem:id-sep}, the probability that there are two identifiers in $S$ that are within $D$ of each other is at most~$n^{-C'}$.

Finally, for graphs in which an adversary makes fewer than $d$ modifications per vertex, the total number of edges the adversary can modify is $r \leq dn/2$. Since all vertices are high- and medium-degree vertices in this model, $N = n$. Therefore, $\ell r/N^2 =O(1/n^{(1-\epsilon)}) \to 0$. \Cref{eq:erdos-t} guarantees that the hypothesis given by \cref{eq:convergence} of \cref{lem:guess} is satisfied. Consequently, the probability that $A$ changes one or more adversary edges is $O(\exp[cn^{1-\epsilon}])$ for some constant $c$.

This proves that each of the three conditions listed above fails with polynomially negligible probability, which implies that the conditional probability is also polynomially negligible.
\end{proof}

\begin{theorem}[Security against an arbitrary edge-flipping adversary in the random power-law graph model]
\label{thm:power-det-adv}
Let $5/2 < \gamma < 3$, $C>0$, $\hm = \hm(\epsilon_1, C_1)$ and $\ml = \ml(\epsilon_1, C_1, \epsilon_2, C_2)$ where $\epsilon_1 = 1$, $\epsilon_2 = 8(C+1)$ and $C_1 = C_2 = C$.

Let $p = P\Square{\ml,\ml}$. Suppose the similarity measure is a vector of distances $\dist = (\dist_e, \dist_v)$, that the corresponding similarity threshold is the vector $\theta =  (r, \log n)$ where $r = p(\ml)^2/32$ is the maximum number of edges the adversary can flip in total, and $\log n$ the maximum number number of edges it can flip per vertex. Suppose that we have $k = n^{C''}$ watermarked copies of the graph, that we use $\ell = 8(2C''+C')(\log n)/p$ to watermark a graph. 

Suppose also that the identifiers map to sets of edges of a graph constrained by the fact that no more than $t = \log n$ edges can be incident to any vertex. Then the watermarking scheme defined in \cref{alg:scheme} is $(G(\vector{w}^\gamma), \dist = (\dist_e, \dist_v), \theta = (r, \log n), k, \ell)$-secure against any deterministic adversary.
\end{theorem}

\begin{proof}
\ifFull
\proofThmPowerDetAdv
\fi
\end{proof}

\subsection{Discussion}
It is interesting to note how the differences in the two random graph models translate into differences in their watermarking schemes. The Erd\H{o}s-R\'{e}nyi model, with its uniform edge probability, allows for constant separation of high-degree vertices, at best. But all the vertices tend to be well-separated. On the other hand, the skewed edge distribution that is characteristic of the random power-law model allows high-degree vertices to be very well-separated, but a significant number of vertices---the low-degree ones, will not be easily distinguished.

These differences lead to the intuition that virtually all edges in the Erd\H{o}s-R\'{e}nyi model are candidates for use in a watermark, as long as only a constant number of selected edges are incident to any single vertex. Therefore, both our watermarking function and the adversary are allowed an approximately linear number of changes to the graph. \Cref{thm:erdos-det-adv} confirms this intuition with a scheme that proposes $O(n^\epsilon)$ bits for the watermark, and a nearly linear number $O(n)$ bits that the adversary may modify.

In contrast, the number of edges that can be used as part of a watermark in the random power-law graph model is limited by the number of distinguishable vertices, which is on the order of $\ml$ or $O(n^\epsilon)$, where $\epsilon = -\frac{2\gamma^2 - 8  \gamma + 5}{2\gamma - 1}$. 

%
%
%

\section{Experiments}

Although our paper is a foundational complement to the systems work of
\ifFull
Zhao {\it et al.}~\cite{DBLP:journals/corr/ZhaoLZZZ15,Zhao:2015},
\fi
we nevertheless provide in this section the results of a small set 
of empirical tests of our methods, so as to experimentally
reproduce the hypothetical
watermarking security experiment from Algorithm~\ref{alg:security}.
Our experiments are performed on
two large social network graphs,
Youtube~\cite{DBLP:journals/corr/abs-1205-6233} from the SNAP
library~\cite{snap}, and Flickr~\cite{mislove-2007-socialnetworks},
as well as a randomly generated graph drawn from the random power-law
graph model distribution. Table~\ref{tab:statistics} illustrates
the basic properties of the networks. To generate the random power-law
graph, we set the number of nodes to $n=10000$, the maximum degree
to $m=1000$, the average degree to $w=20$, and $\gamma=2.75$.

\begin{table}[b!]
\centering
\caption{Network statistics}
\label{tab:statistics}
\begin{tabular}{|l|r|r|r|r|r|r|}
\hline
Network & \multicolumn{1}{l|}{\# nodes} & \multicolumn{1}{l|}{\# edges} & \multicolumn{1}{l|}{Max.\ degree} & \multicolumn{1}{l|}{Avg.\ degree} & \multicolumn{1}{l|}{Unique degree} & \multicolumn{1}{l|}{Estimated $\gamma$} \\ \hline
Power-law & 10,000 & 94,431 & 960 & 18.89 & 14 & --- \\
Youtube & 1,134,890 & 2,987,624 & 28,754 & 5.27 & 29 & 1.48 \\
Flickr & 1,715,256 & 15,554,181 & 27,203 & 18.14 & 130 & 1.62 \\ \hline
\end{tabular}
\end{table}

\begin{table}[b!]
\centering
\caption{Experiment Parameters}
\label{tab:settings}
\begin{tabular}{|l|r|r|r|r|}
\hline
Network & \multicolumn{1}{l|}{\# high-degree} & \multicolumn{1}{l|}{\# medium-degree} & \multicolumn{1}{l|}{Key size} & \multicolumn{1}{l|}{Marking dK-2 deviation}\\ \hline
Power law & 64 & 374 & 219 & 0.065\\
Youtube & 256 & 113 & 184 & 0.033\\
Flickr & 300 & 5901 & 3250 & 0.002\\ \hline
\end{tabular}
\end{table}

\ifFull
\subsection{Adaptations from the theoretical scheme}
\fi

To adapt our theoretical framework to the rough-and-tumble world of 
empirical realities, we made three modifications to our framework for
the sake of our empirical tests.

First, instead of using the high-degree and medium-degree thresholds
derived from Lemmas~\ref{lem:sep-high} and~\ref{lem:sep-mid},
for the power-law distribution, to define
the cutoffs for high-degree and medium-degree vertices, we
used these and the other lemmas 
given above as justifications for the existence of such distinguishing
sets of vertices and we then 
optimized the number of high- and medium-degree vertices to be values that
work best in practice. 
The column, ``Unique degree,'' 
from Table~\ref{tab:statistics} shows, for each network,
the number of consecutive nodes with unique degree when considering
the nodes in descending order of degree. 
\ifFull
This is, in theory, the
maximum number of high-degree nodes that could be distinguished.
\fi
Since this value is too small in most cases, we applied the principles
of Lemmas~\ref{lem:erdos-sep} and~\ref{lem:sep-high} again,
in a second-order fashion,
to distinguish and order the high-degree nodes. 
In particular, in addition to the degree of each high-degree vertex,
we also label each vertex with the list of degrees of its neighbors, sorted in
decreasing order. With this change, we are not restricted in our
choice of number of high-degree nodes as required by applying
these lemmas only in a first-order fashion. Table~\ref{tab:settings}
shows the values used in our experiments based on this second-order
application. As medium-degree vertices, we picked the
maximum number such that there are no collisions among their bit
vectors of high-degree node adjacencies.

Second, instead of returning failure if (a) 
two high-degree nodes have the same degree and list of degrees 
of their neighbors, (b) two medium-degree nodes have the same bit vector, or 
(c) the approximate isomorphism is not injective, we instead 
proceed with the algorithm. 
Despite the existence of collisions, the remaining nodes often provide enough information to conclude successfully.

Finally, 
we simplified how we resampled (and flipped) edges in order to create
a graph watermark, using our approach for the
Erd\H{o}s-R\'{e}nyi model even for power-law graphs,
since 
resampling uniformly among our small set of marked edges is likely not to
cause major deviations in the graph's distribution
and, in any case,
it is empirically difficult to determine the value of $\gamma$ for
real-world social networks.  Therefore, we set the
resampling probability to $0.5$ so that it is consistent with the
Erd\H{o}s-R\'{e}nyi model and so that each bit in the message is
represented uniformly and independently.

\ifFull
\subsection{Fitting real-world networks to the random power-law graph model}

Note that during the marking step, resampling the edges of the key requires 
knowledge of the distribution from which the network is drawn. For this reason, 
we tried to fit the real-world networks to the random power-law graph model 
distribution. The main task was to find the exponent $\gamma$ of a power law 
function that would best fit the degree distribution.

First, we give definitions to help introduce the problem.  The degree distribution of an 
undirected graph $G$ with $n$ vertices is a probability distribution such that 
the probability mass function $P_G$ is given by $P_G(k) =  n_k/ n$ for 
$k \in \{0, \ldots, n\}$ where $n_k$ is the number of vertices with degree $k$.  
A random variable is said to follow a power-law distribution if the probability 
density function is given by $f(x) = cx^{-\gamma}$ for some $c > 0$ and 
$\gamma > 0$.  As stated before, it has been found empirically that the degree 
distributions of many naturally-occurring graphs in social networking and 
science follow a power-law distribution.  Note that for observed data that is 
believed to follow a power-law distribution, the power-law behavior of the data 
often only holds for values larger than some $x_{\min}$ 
(see~\cite{doi:10.1137/070710111}).

In finding $\gamma$, there are three primary focuses: obtaining $\gamma$ 
itself, finding the value for $x_{\min}$ for which the power-law behavior 
holds after, and finding the associated p-value indicating how good of a fit 
the power-law distribution is to the degree distribution of $G$. The methods 
to obtain all three values are taken from 
Clauset {\it et al.}~\cite{doi:10.1137/070710111} and some of the code used 
to find all three values can be found at~\cite{plpval, plfit}. 

To obtain $\gamma$, we use the method of maximum 
likelihood.  The maximum-likelihood estimator for $\gamma$ is given by 
\begin{equation}\label{eq:estimate-gamma}
\gamma = 1 + k\left(\sum_{i=1}^k\ln\frac{x_i}{x_{\min}}\right)^{-1}
\end{equation}
where $x_i$ for $i = 1, \ldots, k$ are the values of the degree distribution
of $G$ such that $x_i \ge x_{\min}$ 
(see~\cite{doi:10.1137/070710111, muniruzzaman1957measures}).

Next, we discuss finding $x_{\min}$. $x_{\min}$ is the lower bound for 
which the data above it follow a power-law distribution.  The idea used to 
find $x_{\min}$ is to choose the value in the observed data (some $x$ 
such that $x \in \{P_G(k) : k = 0,\ldots, n\}$) such that the empirical 
cumulative distribution function (CDF) of the observed data above 
$x_{\min}$ is most similar to the estimated CDF of the power-law 
distribution of the observed data obtained by using Equation~\ref{eq:estimate-gamma}  
(see~\cite{clauset2007frequency}).  We measure similarity between 
distributions using the Kolmogorov-Smirnov (KS) statistic.  So in order to 
find $x_{\min}$, for each $x_i$ we set $x_{\min} = x_i$ and compute the 
KS statistic for the empirical CDF of the observed data larger than 
$x_{\min}$ and the CDF of the power-law distribution obtained by using 
Equation~\ref{eq:estimate-gamma}. The $x_i$ that minimizes the KS 
statistic will be used as the value for 
$x_{\min}$~\cite{doi:10.1137/070710111}.

Last, we discuss finding the p-value based off of the methods 
in~\cite{doi:10.1137/070710111}.  To find the p-value indicating how well 
the power-law distribution fits the degree distribution of $G$, we use the 
methods described above to find $\gamma$ and $x_{\min}$.  After doing 
so, we compute the KS statistic for the empirical CDF of the data and 
the CDF of the power-law distribution with exponent $\gamma$ and 
$x \ge x_{\min}$. We then generate new data sets from our observed 
data.  Let $n_{\text{tail}}$ be the number of observed data larger than 
$x_{\min}$.  Then with probability $n_{\text{tail}}/n$, we randomly sample 
a point from the power-law distribution with exponent $\gamma$ and 
$x \ge x_{\min}$ then add it to the new data set. With probability 
$1 - n_{\text{tail}}/n$, we randomly sample an observed data point in the 
interval $x < x_{\min}$ and add it to the new data set.  This process is 
continued until we have added $n$ total data points.  Then we compute 
the KS statistic for the empirical CDF of the newly generated data set 
and the CDF of the estimated power-law distribution using this newly 
generated data set.  We generate 10,000 data sets and compute the 
KS statistic for each, where 10,000 is a good rule of thumb in order to 
have high precision in the p-value.  The p-value is then determined by 
the fraction of the number of generated KS statistic values that are larger 
than the KS statistic obtained from the original data.  A p-value larger 
than 0.1 implies that the power-law with exponent $\gamma$ and 
$x \ge x_{\min}$ is a good fit to our data~\cite{doi:10.1137/070710111}.

We gathered graphs from many different domains to test if the power-law
distribution is a good fit to the degree distribution.  We tested social networks
from Google+~\cite{mcauley2012learning}, LiveJournal~\cite{mislove-2007-socialnetworks},
Slashdot~\cite{leskovec2009community}, Epinions~\cite{richardson2003trust}, 
Pokec~\cite{takac2012data}, and Twitter~\cite{mcauley2012learning}.
We also tested citation networks of U.S. patents~\cite{leskovec2005graphs} and 
ArXiv~\cite{leskovec2005graphs, gehrke2003overview}, collaboration networks of 
ArXiv~\cite{leskovec2007graph} and the DBLP computer science 
bibliography~\cite{DBLP:journals/corr/abs-1205-6233}, email communication 
networks from the Enron corpus~\cite{leskovec2009community, klimt2004introducing}
and a European research institution~\cite{leskovec2007graph}, a communication 
network from Wikipedia's talk pages~\cite{leskovec2010signed}, location-based online social networks 
(OSN)~\cite{cho2011friendship}, Internet autonomous systems (AS)
networks~\cite{leskovec2005graphs, oregon-RVProject, caida-ASgraph, caida-Skitter}, 
snapshots of the graph of the peer-to-peer file sharing service 
Gnutella~\cite{leskovec2007graph, ripeanu2002mapping}, road networks
from Pennsylvania, Texas, and California~\cite{leskovec2009community}, 
and product co-purchasing networks from 
Amazon~\cite{DBLP:journals/corr/abs-1205-6233, leskovec2007dynamics}.  Additionally, 
we tested web graphs where the nodes represent web pages 
and the edges are hyperlinks connecting the pages for part of Stanford's 
website~\cite{leskovec2009community}, Notre Dame's 
website~\cite{leskovec2009community}, Stanford and Berkeley's 
websites~\cite{leskovec2009community}, and Google~\cite{leskovec2009community}. 
Our results, shown in Table~\ref{tab:gamma-values}, display that the power-law distribution is a good fit for the
degree distribution for 29 of the 40 graphs tested.  Moreover, 23 of the 29 cases 
where it is a good fit, the estimated $\gamma$ is less than $2$.  This prevented
us from using Equation~\ref{eq:powerlaw-probs} for the resampling
probabilities, which requires $\gamma>2$.

\afterpage{%
\begin{table}[b!]
\centering
\caption{Estimated $\gamma$ and associated p-values}
\label{tab:gamma-values}
\begin{tabular}{|c|l|r|r|r|r|}
\hline
Network Type & Network & \multicolumn{1}{l|}{\# nodes} & \multicolumn{1}{l|}{\# edges} & \multicolumn{1}{l|}{Estimated $\gamma$} & \multicolumn{1}{l|}{p-value}\\ \hline
\multirow{7}{75pt}{Social \\Networks} & Google+ & 107,614 & 12,238,285 & 1.8129 & 0.0096 \\
 & LiveJournal & 5,203,764 & 48,709,773 & 1.3901 & 0.0171 \\
 & Slashdot 08/11 & 77,360 & 469,180 & 1.5947 & 0.6425 \\
 & Slashdot 02/09 & 82,168 & 504,230 & 1.5906 & 0.3517 \\
 & Epinions & 75,879 & 405,740 & 1.6893 & 0.9979 \\
 & Pokec & 1,632,803 & 22,301,964 & 2.4144 & 0.8383 \\
 & Twitter & 81,306 & 1,342,296 & 1.4844 & 0.0000 \\ \hline
 
\multirow{3}{75pt}{Citation\\Networks}  & ArXiV (Physics) & 34,546 & 420,877 & 1.3931 & 0.0000 \\
 & ArXiV (Phy. Theory) & 27,769 & 352,285 & 1.4309 & 0.0000 \\
 & Patents & 3,774,768 & 16,518,947 & 1.2776 & 0.2102 \\ \hline
 
\multirow{4}{75pt}{Collaboration\\Networks}  & ArXiV (Astro.) & 18,771 & 198,050 & 2.0062 & 0.3105 \\
 & ArXiV (Condensed) & 23,133 & 93,439 & 1.3896 & 0.0715 \\
 & ArXiV (Physics) & 12,006 & 118,489 & 1.8967 & 0.5440 \\
 & DBLP & 317,080 & 1,049,866 & 1.3405 & 0.0530 \\ \hline
 
 \multirow{3}{75pt}{Communication\\Networks} & Email - Enron & 36,692 & 183,831 & 1.6143 & 0.9884 \\
 & Email - Europe & 265,009 & 364,481 & 1.4703 & 0.9998 \\
 & Wikipedia & 2,394,385 & 4,659,565 & 1.5232 & 0.9933 \\ \hline
 
\multirow{2}{75pt}{Location-based\\OSNs}  & Brightkite & 58,228 & 214,078 & 1.5403 & 0.4212 \\
 & Gowalla & 196,591 & 950,327 & 1.4542 & 0.3132 \\ \hline
 
\multirow{4}{75pt}{AS Networks}  & CAIDA & 26,475 & 53,381 & 1.5446 & 0.9191 \\
 & Skitter & 1,696,415 & 11,095,298 & 1.4484 & 0.7464 \\
 & Oregon - 1 & 11,174 & 23,409 & 1.5523 & 0.9113 \\
 & Oregon - 2 & 11,461 & 32,730 & 1.6024 & 0.8396 \\ \hline
 
\multirow{5}{75pt}{Peer-to-peer\\Networks}  & Gnutella 08/04/02 & 10,876 & 39,994 & 2.5995 & 0.6346 \\
 & Gnutella 08/24/02 & 26,518 & 65,369 & 2.0109 & 0.7580 \\
 & Gnutella 08/25/02 & 22,687 & 54,705 & 2.0495 & 0.6715 \\
 & Gnutella 08/30/02 & 36,682 & 88,328 & 1.8832 & 0.8464 \\
 & Gnutella 08/31/02 & 62,586 & 147,892 & 1.9131 & 0.9071 \\ \hline
 
\multirow{5}{75pt}{Co-purchasing\\Networks} & Amazon 03/02/03 & 262,111 & 899,792 & 1.3310 & 0.6295 \\
 & Amazon 03/12/03 & 400,727 & 2,349,869 & 1.3680 & 0.5020 \\
 & Amazon 05/05/03 & 410,236 & 2,439,437 & 1.3580 & 0.5129 \\
 & Amazon 06/01/03 & 403,394 & 2,443,408 & 1.3684 & 0.6091 \\
 & Amazon (2012) & 334,863 & 925,872 & 1.3382 & 0.5789 \\ \hline

\multirow{3}{75pt}{Road\\Networks} & California & 1,965,206 & 2,766,607 & 1.1415 & 0.0785 \\
 & Pennsylvania & 1,088,092 & 1,541,898 & 3.1585 & 0.6290 \\
 & Texas & 1,379,917 & 1,921,660 & 1.1521 & 0.0642 \\ \hline
 
\multirow{4}{75pt}{Web\\Graphs} & Berkeley \& Stanford & 685,230 & 6,649,470 & 1.4591 & 0.0087 \\
 & Google & 875,713 & 4,322,051 & 1.3890 & 0.3828 \\
 & Notre Dame & 325,729 & 1,090,108 & 1.6577 & 0.8434 \\
 & Stanford & 281,903 & 1,992,636 & 1.4052 & 0.0036 \\ \hline
\end{tabular}
\end{table}
\clearpage}

%
\fi

\subsection{Experiment parameters}

For the experiment parameters other than the original network and the number of high- and medium-degree nodes, we set the following values.
\begin{description}
\item[Maximum flips adjacent to any given node during marking:] 1.
\item[Key size:] 
We set this to the maximum possible value 
(i.e., the number of high- and medium-degree vertices divided by two,
as shown in Table~\ref{tab:settings}), 
because the numbers of high- and medium- degree nodes are not large. 
This effectively means that every high- and medium-degree node 
has exactly one edge added or removed.
\item[Number of marked graphs:] 10.
\item[Adversary:] We used a time-efficient variation of the \textit{arbitrary edge-flipping adversary}. This adversary selects a set of pairs of nodes randomly, and flips the potential edge among each pair.
\end{description}

\subsection{Results}
We evaluated how much distortion the adversary can introduce before
our method fails to identify the leaked network correctly. For this
purpose, we compared the identification success rate to the amount
of distortion under different fractions of modified edges by the
adversary. To estimate the success rate, we ran the experiment 10
times and reported the fraction of times that the leaked network
was identified correctly. As a measure of distortion, we used the
dK-2 deviation~\cite{Zhao:2015} between the original network and
the version modified by the adversary. The dK-2 deviation is the
euclidean distance between the dK-2
series~\cite{Sala:2010:MGM:1772690.1772778} of the two graphs,
normalized by the number of tuples in the dK-2 series. The dK-2
deviation captures the differences between the joint degree
distributions of the networks, that is, the probability that a
randomly selected edge has as endpoints nodes with certain degrees.
We average the dK-2 deviation among the 10 runs. Figure~\ref{fig:results}
shows the outcome of our experiments. Moreover, Table~\ref{tab:settings}
shows the dK-2 deviation introduced by the marking alone.

\begin{figure}[hbt!]
\centering
\includegraphics[width=0.75\linewidth]{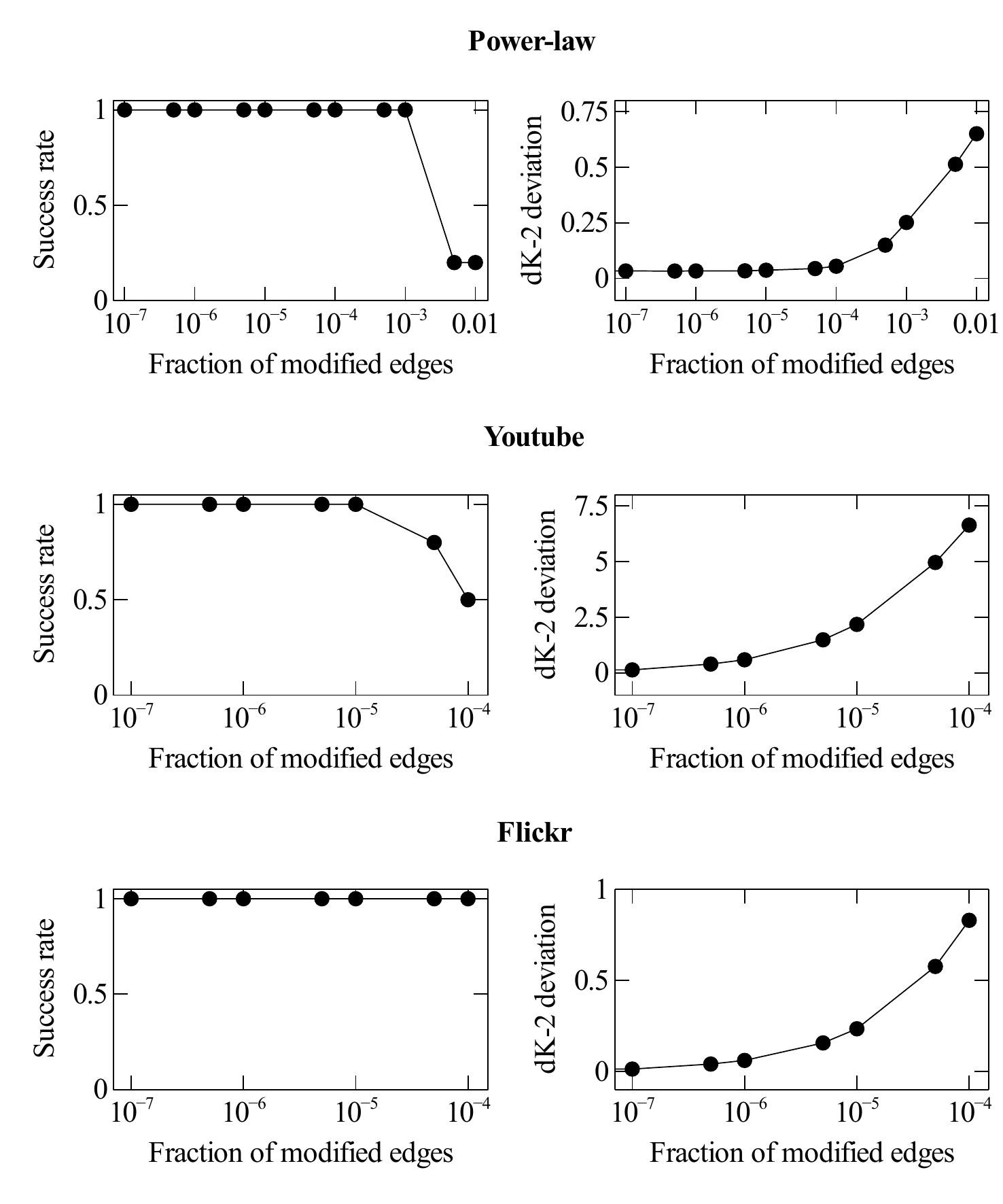}
\vspace*{-8pt}
\caption{Success rate and dK-2 deviation under different fractions
of modified potential edges by the adversary, for the Power law, 
Youtube, and Flickr networks.} \label{fig:results}
\end{figure}

\ifFull
\subsection{Discussion} 
\fi
Based on our experiments, 
the success rate of our scheme is high but it drops after a
certain threshold. This demonstrates that there is a distinct range
of adversarial edge flips that can be tolerated by our scheme.
Specifically, our scheme
worked well when the fraction of potential edges flipped by the adversary is up
to $10^{-3}$ and $10^{-5}$ for the random power-law and Youtube
networks, respectively. For these graphs, this number of flipped
potential edges corresponds to $52.9\%$ and $215.6\%$ of the number of edges
in the original graphs, respectively. For the Flickr network, the
runtime of the adversary modification became excessive before the
success rate could decrease, at a fraction of $10^{-4}$ of potential
edges flipped.

The distortion introduced by the watermark is negligible compared
to the distortion caused by the number of flips that the scheme can
tolerate. On average, the marking modifies half of the edges on the
key, which corresponds to 
$1.1\cdot10^{-3}, 3\cdot 10^{-5},\text{ and }10^{-4}$ of 
the number of edges in the original random power-law,
Youtube, and Flickr networks, respectively.

For the same number of flips, the dK-2 deviation in the Youtube
network was much larger than in the Flickr network, which in turn
was larger than that of the random power-law network. A possible
explanation for this is that any set of uniform edge flips has a bigger
effect on the dk2-deviation of a skewed graph than on the dK-2 deviation
of a less skewed graph. Note that the Youtube network has the largest
skew, as the maximum degree is on the same order as the Flickr
network, but the average degree is less.

\ifFull
\section{Conclusion}
We defined a watermarking framework and a notion 
of security meant to capture the 
difficulty in removing a watermark from a graph. 
We studied two random graph models and showed that watermarking in these 
models could be achieved in such a way that no adversary could remove 
the watermark whp and still have a graph that is ``close'' to his
original graph. 
A vital feature of our approach is that
we watermark graphs to look like typical graphs from 
the distribution the original graph was issued from, 
while also making them look similar to the original.
In addition, we provided an exemplary implementation that works effectively
for marks consisting only of edge flips.

For future work, it would be interesting to consider solutions that
can tolerate some degree of collusion. 
In the exemplary schemes we presented,
an adversary is likely to detect many edges of the watermark, 
if he has access to multiple watermarked graphs produced 
from the same original graph.

\subsection*{Acknowledgments} 
This research was supported in part by
the National Science Foundation under grants 1011840
and 1228639.
This article also reports on work supported by the Defense Advanced
Research Projects Agency (DARPA) under agreement no.~AFRL FA8750-15-2-0092.
The views expressed are those of the authors and do not reflect the
official policy or position of the Department of Defense 
or the U.S.~Government.
\fi

{\raggedright 
\ifFull
\bibliographystyle{abbrv} 
\fi
\bibliography{paper} 
}
\end{document}